\documentclass[10pt,journal]{IEEEtran}

\usepackage{amsmath,amsfonts,amssymb}
\usepackage{amsthm}
\usepackage{bm}
\usepackage{booktabs}
\usepackage{caption}
\usepackage[scale=2]{ccicons}
\usepackage{cite}
\usepackage{color}
\usepackage{etoolbox}
\usepackage{graphicx}
\usepackage{lettrine}
\usepackage{mathtools}
\usepackage{multirow}
\usepackage{pgfplots}
\usepackage{ragged2e}
\usepackage{rotating}
\usepackage{subfigure}
\usepackage{xspace}

\usepackage[font=small,skip=-1pt]{caption}

\newtheorem{mydef}{Definition}
\newtheorem{prop}{Proposition}

\ifCLASSINFOpdf
\else
 \fi
\begin{document}
	\title{Ruin Theory for Dynamic Spectrum Allocation in LTE-U Networks}
%
	
	\author{
		\IEEEauthorblockN{Aunas Manzoor, Nguyen H. Tran, Walid Saad, S. M. Ahsan Kazmi, Shashi Raj Pandey, and Choong Seon Hong}		
}
		
%
\makeatletter
\patchcmd{\@maketitle}
{\addvspace{0.5\baselineskip}\egroup}
{\addvspace{-1.5\baselineskip}\egroup}
{}
{}

\makeatother

	\maketitle
\begin{abstract}
	LTE in the unlicensed band (LTE-U) is a promising solution to overcome the scarcity of the wireless spectrum. However, to reap the benefits of LTE-U, it is essential to maintain its effective coexistence with WiFi systems. Such a coexistence, hence, constitutes a major challenge for LTE-U deployment. In this paper, the problem of unlicensed spectrum sharing among WiFi and LTE-U system is studied. In particular, a fair time sharing model based on \emph{ruin theory} is proposed to share redundant spectral resources from the unlicensed band with LTE-U without jeopardizing the performance of the WiFi system. Fairness among both WiFi and LTE-U is maintained by applying the concept of probability of ruin. In particular, the probability of ruin is used to perform efficient duty-cycle allocation in LTE-U, so as to provide fairness to the WiFi system and maintain certain WiFi performance. Simulation results show that the proposed ruin-based algorithm provides better fairness to the WiFi system as compared to equal duty-cycle sharing among WiFi and LTE-U.  
\end{abstract}
\vspace*{-0.5cm}
\IEEEpeerreviewmaketitle
\vspace*{-0.6cm}
	\section{Introduction}
	\lettrine[findent=2pt]{\textbf{T}}{he} unprecedented rise in the demand for wireless data traffic is constantly straining the capacity of existing cellular networks 
	thus yielding major challenges related to spectrum scarcity \cite{kazmi2016optimized}. 
	In order to overcome this spectrum scarcity, the idea of operating wireless cellular networks in the unlicensed band, using the so-called LTE in unlicensed band (LTE-U) technology has recently been proposed \cite{7143339}.

	To effectively reap the benefits of LTE-U deployment, cellular network operators must ensure an effective co-existence between their LTE-U users and incumbent WiFi systems \cite{zhou2017licensed}. Guaranteeing an effective LTE-U and WiFi coexistence is challenging due to the discrepancies in medium access mechanisms for both WiFi and LTE-U. For instance, LTE-U is more spectrum efficient due to its centralized MAC scheduling protocol as compared to WiFi systems. However, LTE-U can cause WiFi performance degradation in the form of long delays and additional collisions if no fair mechanism is applied for sharing the unlicensed spectrum. 
	
	Several approaches have been recently proposed to ensure effective LTE-U and WiFi coexistence \cite{chen2017echo,ko2016fair,CU-LTE,bairagi2017lte,challita2017proactive}. 
	 Such approaches include listen-before-talk (LBT) \cite{ko2016fair} in which clear channel assessment (CCA) is used before allocating the spectrum to LTE-U. 
	 However, CCA becomes challenging to implement in a dynamic WiFi environment. Another common spectrum access mechanism in LTE-U is the so-called duty cycle based spectrum access in which a channel is accessed by LTE-U users for some time duration and then released for WiFi systems \cite{CU-LTE}. In \cite{bairagi2017lte}, a Nash bargaining game is proposed to maximize the sum rate of cellular users by exploring the unlicensed band. 
	However, the problem of efficient duty-cycle allocation for the coexistence of LTE-U and WiFi system remains a major challenge. Moreover, due to the uncertain behavior of WiFi systems,  efficient duty-cycle management in LTE-U is a significant open problem that needs to be addressed.  
	
The main contribution of this paper is, thus, to introduce a new framework for enabling a fair and efficient coexistence between LTE-U and WiFi systems. The proposed approach allows maximizing the rate of cellular systems by efficiently utilizing the unlicensed spectrum while providing sufficient throughput for WiFi. To develop such a fair coexistence mechanism, we propose a novel approach based on \emph{ruin theory}\cite{davis2005insurance}. The use of ruin theory for managing LTE-U and WiFi coexistence is apropos because ruin theory allows maintaining a certain performance guarantee for the of WiFi system via the notion of a ruin probability. This probability of ruin depends on two key factors: a) the  collisions in the WiFi system and , b) the proportion of duty cycle allocated to LTE-U. In summary, the main contributions of this paper include: 
	\begin{itemize}		
	\item We formulate the problem of allocating redundant spectrum resources from the unlicensed band to LTE-U users. The objective is to maximize the LTE-U rate while incorporating sufficient transmission opportunities for the WiFi system using chance-constrained optimization.  
	\item We use ruin theory to model the surplus process for WiFi duty-cycle. The surplus WiFi duty-cycle is then used to find the probability of ruin. 
	\item Moreover, a probability of ruin based optimization problem is formulated and solved to dynamically allocate the available unlicensed spectrum to LTE-U users.  
	\item Simulation results show that the LTE-U duty-cycle proportion decreases as the probability of ruin increases. The results also show that the proposed ruin-theoretic approach provides better fairness to the WiFi system as compared to a baseline equal duty-cycle sharing scheme among WiFi and LTE-U.      		
	\end{itemize}
	
	 To our best knowledge, this is the first work that adopts ruin theory for managing LTE and WiFi coexistence.
			
\vspace*{-0.4cm}
	\section{System Model and Problem Formulation}
	\label{sec-systemModel}	
	\vspace*{-0.1cm}
	\subsection{System Model}
	\vspace*{-0.1cm}
	Consider the downlink of a wireless cellular network consisting of one LTE small-cell base station (SBS) and a set $\mathcal{U}$ of $U$ user equipment (UE) coexisting with a WiFi system composed of a set $\mathcal{W}$ of $W$ WiFi access points (WAPs). For each WAP $w \in \mathcal{W}$, there is a corresponding set $\mathcal{U}_w$ of $U_w$ WiFi stations (WSTs). The SBS operates using a licensed-assisted access (LAA) protocol using which the unlicensed spectrum is used only for downlink traffic while uplink and other control traffic will use the licensed spectrum. WAPs use the DCF mode of the wireless LAN (WLAN) IEEE 802.11 protocol. The SBS and all of the WAPs are deployed in the same geographical area 
	and use the same unlicensed spectrum.  
		
	Multiple WAPs coexist with an SBS on the unlicensed band, and, hence, the SBS transmissions may interfere with WiFi transmissions. The WAPs will use non-overlapping channels while the SBS opportunistically accesses the unused spectrum from all of these WAPs. 
	We consider a set $\mathcal{K}$ of $K$ non-overlapping unlicensed spectrum channels. 
	
	To share the unlicensed spectrum among WiFi and LTE-U, we consider two types of time frames: A long frame and a short frame. For each channel $k$, a long frame $T$  is divided into a set $\mathcal{N}$ of $N$ short frames  each of which having a duration $\delta = 1$~ms as defined by 3GPP \cite{3gpp211}. Fig. \ref{frame} shows the long frame distribution for an arbitrary unlicensed channel that can be accessed by either the WAPs or the SBS. There are three possibilities in each short frame $n \in \mathcal{N}$ of channel $k \in \mathcal{K}$: (a) Successful WiFi transmission, (b) WiFi collision, or (c) SBS access represented as LTE-U slot as shown in Fig. \ref{frame}. 
	
	\begin{figure}[!t]
		\centering			
		\includegraphics[width=0.4\textwidth]{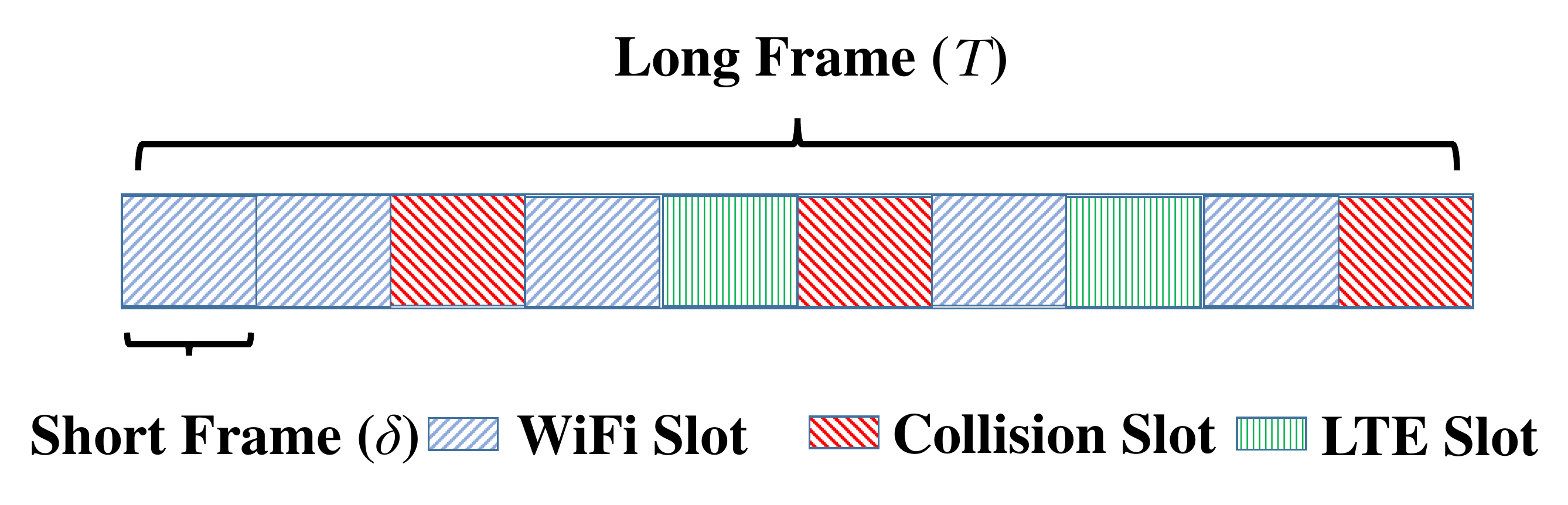}
		\caption{Frame structure for channel access. }
		\label{frame}
	\end{figure}
	
	Collisions can happen only in WiFi systems when more than one WAP tries to transmit in the same time slot. As WiFi systems sense the channel availability for a $DIFS$ period of time before transmission, no WiFi system will transmit if the SBS starts transmitting during this $DIFS$ period. On the other hand, the SBS can start transmitting when the channel is free without waiting for the $DIFS$ period. Therefore, no collision will happen during LTE-U transmission, and thus LTE-U may always access the channel. This can result in significant performance degradation of WiFi in terms of long delays and increased collisions. Therefore a fairness mechanism is required to address the coexistence issue between LTE-U and WiFi by limiting the LTE-U duty-cycle allocation. 
	
	In order to provide such fairness to the WiFi system, we allocate a fixed number of short frames $r$ out of $N$, each of duration $\delta$, to WAPs operating on channel $k$. Essentially, in every long frame $T$, there are $N$ short frames out of which $r$ short frames are allocated to WiFi and the rest can be allocated to LTE-U opportunistically while taking into account the random WiFi collisions $X$. LTE-U frames are further divided among cellular user equipments (UEs). The typical $\log$ utility for each UE $i\in \mathcal{U}$ is given by: 
	\vspace{-0.2cm}
		\begin{equation}
		\gamma_{ik} = \log \left(1 +  \frac{P_i g_{ik}}{\sigma^2} \right), 
		\end{equation}
	where $P_i$ and $g_{ik}$ represent, respectively the downlink power and channel path-loss gain for UE $i\in\mathcal{U}$ on channel $k \in \mathcal{K}$.
 	
	For every channel $k \in \mathcal{K}$, the WiFi system can experience a random number of collisions which can be represented by a random variable $X$ that follows a Poisson distribution \cite{article} with arrival rate $\lambda_k$. To maintain the WiFi system throughput, we need to ensure that there are at-least $r$ number of short frames for the successful transmissions in the WiFi system. This can be guaranteed by using a chance constraint to opportunistically choose $\alpha_{kn} \in \{0,1\}$ such that the desired WiFi performance is ensured: 
	\begin{equation}
	\Pr \left[X + \sum_{n \in \mathcal{N}}  \alpha_{kn} \leq N-r \right] \geq \xi,\quad \forall k \in \mathcal{K}.
	\label{eq:chance} 
	\end{equation}
	In (\ref{eq:chance}), $\alpha_{kn}$ indicates whether the short frame $n \in \mathcal{N}$ of channel $k \in \mathcal{K}$ will be allocated to LTE-U. The constraint in (\ref{eq:chance}) ensures that the probability of allocating a sufficient number, $r$, of short frames to the WiFi system is higher than a threshold $\xi$. 
\vspace*{-0.5cm}
	\subsection{Problem Formulation}
	\label{secProbForm}
	Given the defined system model, our goal is to maximize the rate of LTE-U while ensuring a guaranteed spectrum access opportunity for WiFi system. To achieve this goal, we formulate the following optimization problem: 
\vspace*{-0.2cm}	
	\begin{subequations}\label{eq:obj1}
	\begin{align}
	\max_{\boldsymbol{y},\boldsymbol{\alpha}} 
	& \sum_{k \in \mathcal{K}} \sum_{n \in \mathcal{N}} \sum_{i \in \mathcal{U}} \alpha_{kn} \log \left( 1 + y_{kn}^{(i)} \gamma_{ik} \right)  \tag{\ref{eq:obj1}} \\
	\text{s.t.} \quad &\label{eq:p1const1} \Pr \left[X + \sum_{n \in \mathcal{N}}  \alpha_{kn} \leq N-r \right] \geq \xi ,\quad \forall k \in \mathcal{K},\\ 
	& \label{eq:p2const2_} \sum_{i \in \mathcal{U}}  y_{k}^{(i)} \leq B \alpha_{kn}, \quad  \forall  k \in \mathcal{K},	\\
	&\label{eq:p1const3} y_{kn}^{(i)} \geq 0 , \quad \forall i \in \mathcal{U}, k \in \mathcal{K}, n \in \mathcal{N},	\\ 
	&\label{eq:p1const4} \alpha_{kn}	\in \{0,1\},\quad \forall n \in \mathcal{N}, k \in \mathcal{K}, 		
	\end{align}
\end{subequations}
where $B$ is bandwidth of unlicensed spectrum, and $y_{kn}^{(i)}$ is the bandwidth of channel $k$ allocated to user $i$ in short frame $n$. The objective function in (\ref{eq:obj1}) captures the rate of LTE-U. The constraint in (\ref{eq:p1const1}) limits the allocation of short frames to LTE-U below a threshold to give sufficient spectrum access opportunity to WiFi. 
(\ref{eq:p2const2_}) is the constraint for bounded channel allocation to LTE-U users. (\ref{eq:p1const3}) and (\ref{eq:p1const4}) are bounds for the decision variables. 

The randomness in the optimization problem due to $X$  in constraint (\ref{eq:p1const1}) and the binary variable $\alpha_{kn}$ make the problem challenging to solve. 
As such, in the next section, we propose to solve this problem by applying \emph{ruin theory} \cite{davis2005insurance}. 

\vspace*{-0.4cm}
	\section{Ruin Theory: Preliminaries}
	\label{secRuin}
	\vspace*{-0.1cm}
	 In the economics literature \cite{davis2005insurance}, ruin theory is used to model a \emph{surplus process} is defined as the process of variations in the capital of an insurance company over time. An important metric is the so-called \emph{probability of ruin} which essentially represents the probability of getting a negative surplus. We will address our problem of fair WiFi/LTE-U coexistence using \emph{ruin theory} by introducing a surplus process for the WiFi/LTE-U duty-cycle. The adoption of ruin theory as compared to other approaches like game theory, is suitable for such problems since ruin theory can efficiently capture the random performance of the WiFi system. Based on the probability of ruin of WiFi, the spectrum resources will be allocated to LTE-U.  
	
	 \label{seSurp}
	 The \emph{WiFi surplus process} is modeled using a discrete-time risk process for the WiFi duty-cycle in the long frame. This WiFi surplus is composed of three main components including: a) initial capital, b) premium rate, and c) random claims. 
	 \begin{mydef}
	 	The \emph{initial capital} represented as $u$ is defined as the initial surplus of insurance at time $0$. 
	 		 
	 	The \emph{premium rate} represented as $r$ is defined as the constant income of insurer received in unit time. 
	
	 	The \emph{claims} are randomly distributed and defined as the request to the insurance company for the compensation of loss.   
	 \end{mydef}
	 The WiFi surplus i.e.\ WiFi duty-cycle, is dependent on the number of fixed short frames $r$ allocated to WiFi, the number of random collisions happening in the WiFi system, and the short frames accessed by LTE-U, i.e., $\sum\nolimits_{n \in \mathcal{N}}\alpha_{kn}$. We relax the binary variable $\alpha_{kn} \in \{0,1\}$ to the continuous domain as $\alpha_{k} \in [0,T]$ where $\alpha_{k}$ represents proportion of LTE-U duty-cycle for channel $k \in \mathcal{K}$.
	 
	 Let, $S = \sum\nolimits_{i = 1}^{n} X_i$ be a composite random variable representing the total time wasted in collisions during $n$ number of time slots, where $X_i$ is exponentially distributed with parameter $\mu$. LTE-U duty-cycle allocation can be incorporated into $X_i$ to yield another random variable $S' = \sum\nolimits_{i = 1}^{n} Z_i$ where $Z_i = X_i + \alpha_{k}$  represents the depreciation in WiFi surplus duty-cycle by collision time and LTE-U duty-cycle. The random variable $Z_i$ is exponentially distributed with parameter $\mu' = \mu + \alpha_{k}$ where $\alpha_{k}$ is a deterministic LTE-U duty-cycle. As a result, we can formally define the WiFi surplus as follows:
	  
\begin{figure}[!t]
	\centering			
	\includegraphics[width=0.4\textwidth]{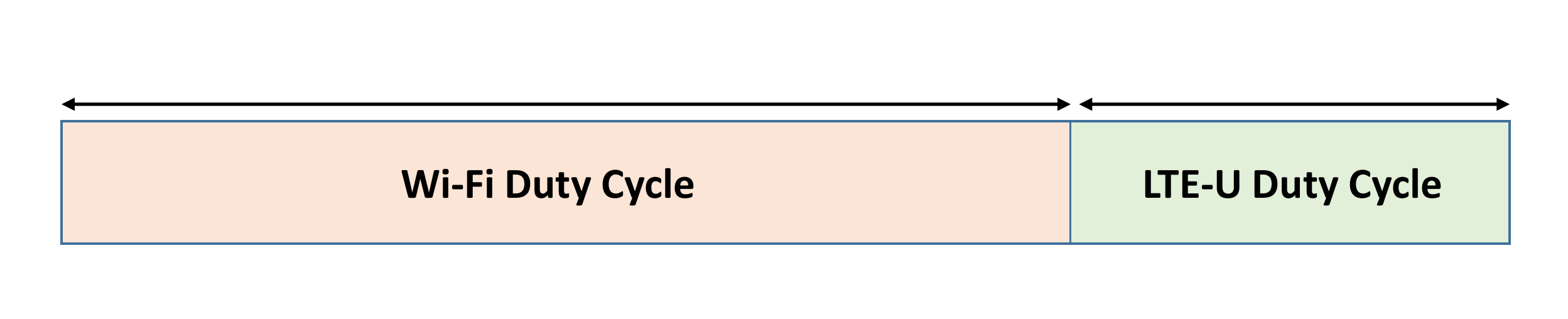}
	\caption{Duty cycle based on the probability of ruin. }
	\label{duty}
\end{figure}
	 
	  
\begin{figure*}	
	\subfigure[Probability of ruin vs. LTE-U duty-cycle proportion.\label{alpha}]{\includegraphics[width=0.31\textwidth]{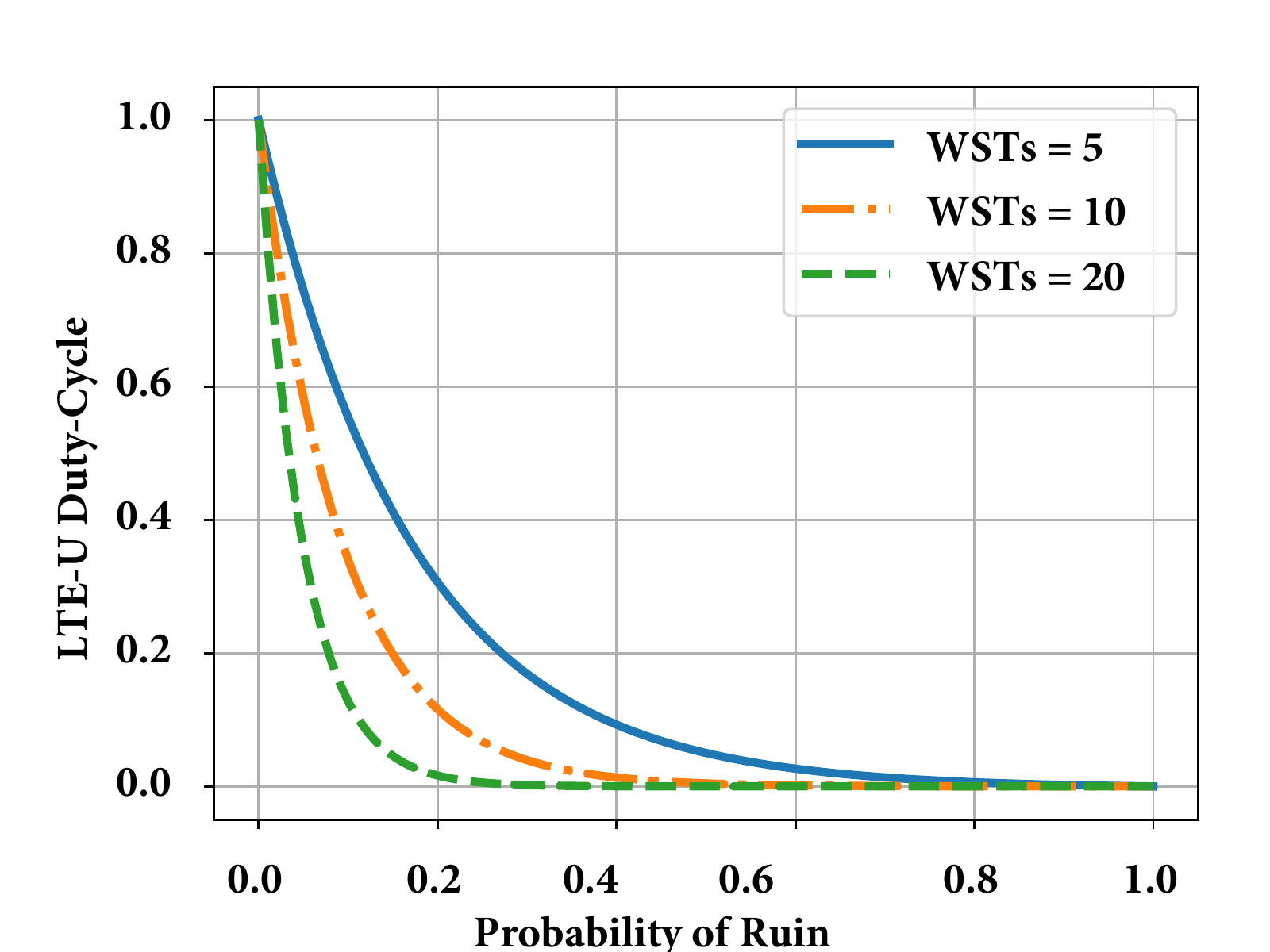}}\quad	
	\subfigure[WiFi throughput plot.\label{wifi}]{\includegraphics[width=0.32\textwidth]{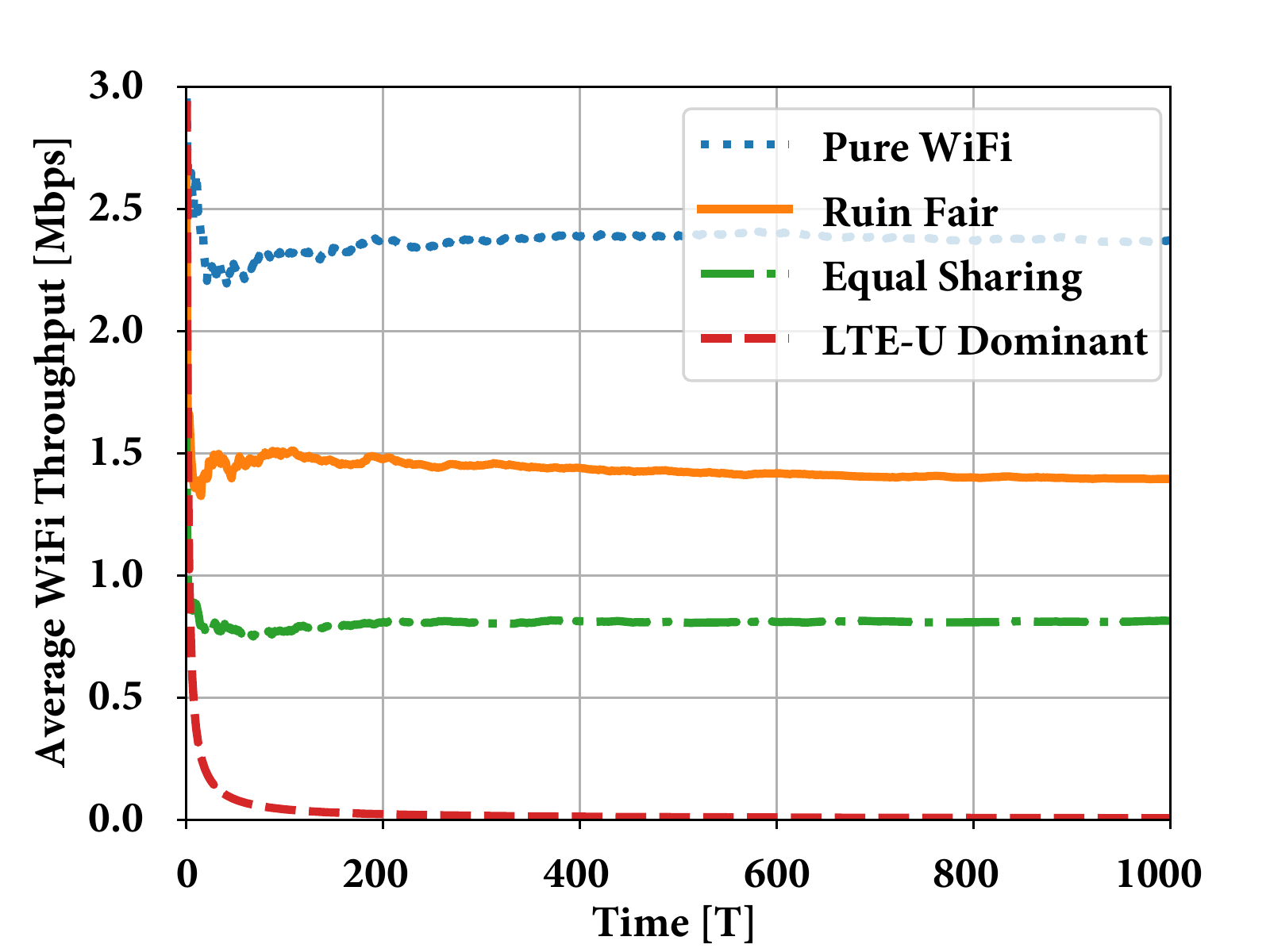}}\quad
	\subfigure[LTE-U sum rate.\label{lte-u}]{\includegraphics[width=0.31\textwidth]{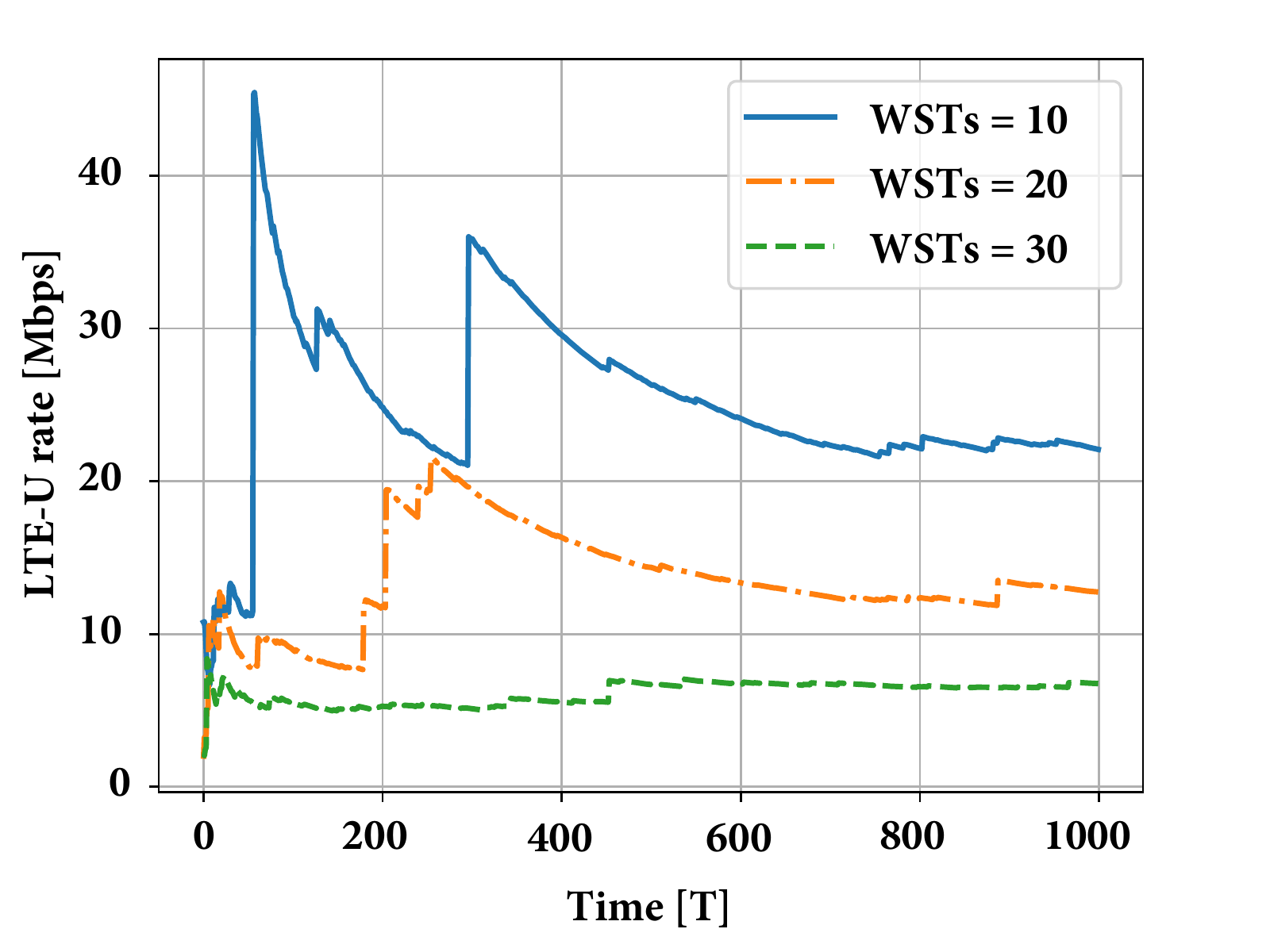}}	
	\caption{\ref{alpha} shows the decrease in LTE-U duty-cycle as the probability of ruin of WiFi increases. \ref{wifi} shows that better fairness to WiFi is provided as compared to other duty-cycle sharing schemes. \ref{lte-u} shows that the proposed fairness mechanism decreases the LTE-U rate when the number of WiFi stations in the network increases.}	
	\label{results}
\end{figure*}

\vspace*{-0.9cm}
		 	\begin{equation}
		 	U_k(n)=u_k + nr\delta - S'.
		 	\label{surplus}
			\end{equation}

		The surplus in (\ref{surplus}), which represents the WiFi duty-cycle, is used to find the finite-time probability of ruin which is defined as follows \cite{finiteruin}.

		\begin{mydef}
			The \emph{finite-time probability of ruin} is defined as the probability of getting a negative surplus at any time instant $s$ in a discrete time $n$. Mathematically, the finite time probability of ruin is represented as: 
	 	\begin{equation}
	 	\psi(u,n) = \textrm{Pr} [U (s) <0 \text{, for some $s$ as } 0 < s< n],
	 	\end{equation} 
	 	where $u$ is the starting surplus, $n$ is the total number of discrete time units, and $s$ ranges from $0$ to $n$. 
	 	\end{mydef}
	 	In the surplus process, the depreciation in surplus is modeled using an exponential distribution of parameter $\mu'$ which represents the distribution parameter of the claims in the surplus process. The finite-time probability of ruin $\psi_k(u_k,n)$ is determined as follows \cite{finiteruin}: 
	 \begin{equation}
	 \psi_k(u_k,n) = \sum_{j=1}^{n} \frac{[\mu'c_j(u_k))]^{j-1}}{(j-1)!} e^{-\mu'c_j(u_k)} \frac{c_1(u_k)}{c_j(u_k)},	 	
	 \end{equation}  
	 where $c_j(u_k) = u_k + jc$, and $c_1(u_k) = u_k +c$. $\psi_k(u_k,n)$ refers to the probability of ruin of the WiFi system at time slot $n$ with $u_k$ denoting the initial WiFi surplus for channel $k \in \mathcal{K}$. 
\vspace*{-0.5cm}	
	\section{Proposed Ruin Theory Based Solution:}
	\label{secSol}
	After finding the probability of ruin, we can use it to provide fairness while allocating the resources to LTE-U. Our goal is to maximize the rate of LTE-U while maintaining sufficient WiFi throughput. A high probability of ruin corresponds to a low WiFi throughput, therefore fewer spectrum resources will be allocated to LTE-U in order to give more spectrum access to WiFi. This is done by choosing the LTE-U spectrum allocation proportion $\alpha_k$ based on the probability of ruin and dividing this spectrum among LTE-U users to maximize the rate of LTE-U. 
	
	\begin{prop}
		Constraint (\ref{eq:p1const1}) can be transformed into an equivalent probability of ruin based constraint $\alpha_k \leq (1-\psi_k(u_k,N))T$.
		\label{prop}
	\end{prop}  
\begin{proof}\nobreak\ignorespaces
	We want to select $\alpha_k$ such that the probability of ruin of WiFi is below some threshold $\tau$ i.e.\ $	(1-\psi_k(u_k,N))T \geq  \tau$, where $\psi_k(u_k,N)$ refers to the probability of getting the surplus $u_k + nc\delta - S'$ negative. Here $S' = X_t + \alpha_k$, where $X_t$ represents the random WiFi collision time. By substituting this in the ruin-based constraint, we obtain:  
	\vspace*{-0.5cm}	
	
	\begin{equation}
	( P[  X_t   \leq u_k + nc\delta - \alpha_k])T \geq \tau \text{,  }\quad \forall  k \in \mathcal{K}.
	\end{equation}
	\vspace*{-0.5cm}	
	
	Let $\eta = \tau / T$. We have $\alpha_k = \sum_{n \in N}\alpha_{kn}$. Moreover, $u_k$ is the initial capital of WiFi. We can consider $u_k = N$, because initially, all of the $N$ slots are available for WiFi transmission. We also consider $r = nc\delta$, where $r$ represents the number of WiFi slots to needed to satisfy the chance constraint $P[X_t + \sum_{n \in N}\alpha_{kn} \leq N - r ] \geq    \eta $. 	
	\end{proof}	
	Note that $B \alpha_k$ indicates that the unlicensed bandwidth $B$ is available to LTE-U for the duration $\alpha_k$. To easily solve the problem (\ref{eq:obj1}), we decompose the formulated problem into two sub-problems. In the first sub-problem, the long frame is adaptively divided into WiFi duty cycle and LTE-U duty cycle using the constraint in Proposition \ref{prop} as shown in Fig. \ref{duty}. The probability of ruin based LTE-U duty cycle solved from the problem (\ref{eq:obj1})  can be represented as:
	\begin{equation}
	\alpha_{k}^{*} = (1-\psi_k(u_k,N))T, \quad  k \in \mathcal{K},
	\end{equation}
	  where $T$ is the total duration of long frame determined as $T = N\delta$ with $N$ being the number of short frames in a long frame and $\delta$ being the duration of each short frame.

	Given $\alpha_{k}^{*}$, the next step is to allocate the available LTE-U duty cycle $\alpha_{k}^{*}$ to LTE-U users. $y_{k}^{(i)}$ represents the proportion of LTE-U duty cycle of channel $k$ to be allocated to user $i$. This resource allocation problem can be formulated as follows: 
\begin{subequations}\label{eq:obj3}
		\begin{align}
		\max_{\boldsymbol{y}} \label{eq:obj3} \quad 
		& \sum_{ k \in \mathcal{K}}  \sum_{i \in \mathcal{U}} \alpha_{k}^{*} \log \left(1 + y_{kn}^{(i)} \gamma_{ik} \right) \\
		 \text{s.t.} \quad & \label{eq:p3const3} \sum_{i \in \mathcal{U}}  y_{k}^{(i)} \leq B \alpha_{k}^{*} , \quad k \in \mathcal{K},\\
		  &\label{eq:p3const4} y_{k}^{(i)} \geq 0 , \quad \forall i \in \mathcal{U}, k \in \mathcal{K} . 
		\end{align}
	\end{subequations}

	The above problem is convex because the objective and all of the constraints are convex. We solved this problem through water-filling and KKT conditions and obtained the following solution:
\vspace{-0.3cm}
\begin{equation}
y_{k}^{*(i)} = \left[\frac{ \alpha_{k}^{*}}{\lambda_k}  - \frac{1}{\gamma_{ik}} \right]^+ , \quad \forall i \in \mathcal{U}, k \in \mathcal{K},
\label{wtfill}
\end{equation}
From \ref{wtfill}, we get $\lambda_k^*$ which is the optimal water level chosen such that the condition
$\sum\nolimits_{i \in \mathcal{U}} y_{k}^{*(i)} = B\alpha_{k}^{*}$ is satisfied with equality for all $k \in \mathcal{K}$. Proportional resource allocation $y_{k}^{*(i)}$ is performed based on $\gamma_{ik}$ to maximize the LTE-U sum rate. The analysis shows that the first sub-problem is computing $\alpha_k$ from the probability of ruin with complexity $\mathcal{O}(N)$ and the second sub-problem is using the water-filling algorithm and has complexity $\mathcal{O}(U.K)$.  
\vspace*{-0.5cm}
\section{Simulation Results}
\vspace*{-0.1cm}
\label{secSim}
For our simulations, we consider one SBS of radius $200$~m that coexists with 3 WAPs that are uniformly distributed with each having a coverage area of $100$~m.  
The WiFi network setup is based on the IEEE 802.11 basic DCF mode (without RTS/CTS). Fig. \ref{alpha} shows that the LTE-U duty-cycle proportion is reduced when the probability of ruin for the WiFi system is high. In fact, no spectrum resources are allocated to LTE-U when $\psi(u,n) >0.4$. Furthermore, with the increase in the number of WiFi stations in the network, the LTE-U duty-cycle is reduced. The main reason is that there is an increase in the possible number of collisions in the WiFi system due to more $WSTs$ contending for channel access causing an increase in the probability of ruin. Therefore, when the number of WSTs are increased in the WiFi network, the LTE-U duty-cycle allocation is further reduced in order to accommodate the additional WSTs.

Fig. \ref{wifi} shows that the proposed ruin-fair and equal sharing approaches achieve up to $62.5\%$ and $29.17\%$ of the throughput of pure WiFi. It can be seen that the ruin-fair approach achieves better WiFi throughput. This is due to the fact that the ruin-based solution restrains LTE-U transmission in case of bad WiFi performance. Moreover, the ruin-based approach can better utilize the idle WiFi channels to improve spectrum efficiency. Conversely, equal sharing and LTE-dominant schemes suppress WiFi performance for better spectrum efficiency. 

Fig. \ref{lte-u} shows the reduction in LTE-U sum rate when the number of WSTs in the network increases which results in more WiFi collisions. We can also see that the LTE-U duty-cycle and, hence, LTE-U rate will decrease when the number of WSTs (and potential collisions) increases. This provides fairness to the WiFi system by prioritizing the WiFi network thus allowing it to maintain its performance.The proposed ruin-based approach can be leveraged in practical coexisting LTE-U and WiFi networks deployed in public hotspot areas. 
\vspace*{-0.2cm}
\section{Conclusion}
\vspace*{-0.2cm}
\label{secConc}
In this paper, we have studied the problem of LTE-U and WiFi coexistence in the unlicensed spectrum. We have formulated the coexistence problem as an optimization problem for LTE-U allocation under the WiFi fairness constraint. This problem is transformed into a ruin-theoretic problem. 
Numerical results show that fairness is provided to WiFi by allocating less resources to LTE-U accordingly whenever, there is performance degradation in the WiFi systems in the form of probability of ruin. Future work can extend this approach to 
accommodate co-existence among other technologies.\\ 
\vspace*{-0.7cm}
\bibliographystyle{IEEEtran}
\bibliography{References}

\end{document}